\documentclass[a4paper,11pt]{article}

\usepackage{fullpage}

\usepackage{algorithm}
\usepackage{algpseudocode}

\usepackage{booktabs} 
\usepackage{lscape}   
\usepackage{cite}

\usepackage{latexsym}
\usepackage{url}
\usepackage{listings}
\usepackage{xspace}
\usepackage{color}

\usepackage{amsmath}
\usepackage{amssymb}
\usepackage{amsthm}

\usepackage{rotating,graphicx}

\newcommand{\ceiling}[1]{\lceil #1\rceil}
\newcommand{\floor}[1]{\lfloor #1\rfloor}

\newcommand{\skips}[1]{\mathtt{skip[}#1\mathtt{]}}

\newcommand{\ablock}{\mathtt{block}}
\newcommand{\sendblock}[1]{\mathtt{sendblock[}#1\mathtt{]}}
\newcommand{\recvblock}[1]{\mathtt{recvblock[}#1\mathtt{]}}

\newcommand{\nextlink}[1]{\mathtt{next[}#1\mathtt{]}}
\newcommand{\prevlink}[1]{\mathtt{prev[}#1\mathtt{]}}

\newcommand{\mpibcast}{\textsf{MPI\_\-Bcast}\xspace}
\newcommand{\mpiallgather}{\textsf{MPI\_\-Allgather}\xspace}
\newcommand{\mpiallgatherv}{\textsf{MPI\_\-Allgatherv}\xspace}
\newcommand{\mpireduce}{\textsf{MPI\_\-Reduce}\xspace}
\newcommand{\mpireducescatterblock}{\textsf{MPI\_\-Reduce\_\-scatter\_\-block}\xspace}
\newcommand{\mpireducescatter}{\textsf{MPI\_\-Reduce\_\-scatter}\xspace}
\newcommand{\mpiint}{\textsf{MPI\_INT}\xspace}

\newcommand{\bidirec}[2]{\textsf{Send}(#1)\parallel\textsf{Recv}(#2)\xspace}

\newtheorem{theorem}{Theorem}

\newtheorem{lemma}{Lemma}
\newtheorem{observation}{Observation}

\newcommand{\gcc}{\texttt{gcc~8.3.0}\xspace}

\newcommand{\hydraopenmpi}{OpenMPI~4.0.5\xspace}

\newcommand{\vegampi}{OpenMPI~4.1.5\xspace}

\title{Optimal Broadcast Schedules in Logarithmic Time with
  Applications to Broadcast, All-Broadcast, Reduction and
  All-Reduction}

\author{Jesper Larsson Tr\"aff\\
  TU Wien\\
  Faculty of Informatics\\
  Institute of Computer Engineering, Research Group Parallel Computing 191-4\\
  Treitlstrasse 3, 5th Floor, 1040 Vienna, Austria}
\date{July 22nd, 2024}

\begin{document}

\maketitle

\begin{abstract}
  We give optimally fast $O(\log p)$ time (per processor) algorithms
  for computing round-optimal broadcast schedules for message-passing
  parallel computing systems.  This affirmatively answers difficult
  questions posed in a SPAA 2022 BA and a CLUSTER 2022 paper.  We
  observe that the computed schedules and circulant communication graph
  can likewise be used for reduction, all-broadcast and all-reduction
  as well, leading to new, round-optimal algorithms for these
  problems. These observations affirmatively answer open questions
  posed in a CLUSTER 2023 paper.

  The problem is to broadcast $n$ indivisible blocks of data from a
  given root processor to all other processors in a (subgraph of a)
  fully connected network of $p$ processors with fully bidirectional,
  one-ported communication capabilities. In this model,
  $n-1+\ceiling{\log_2 p}$ communication rounds are required. Our new
  algorithms compute for each processor in the network receive and
  send schedules each of size $\ceiling{\log_2 p}$ that determine
  uniquely in $O(1)$ time for each communication round the new block
  that the processor will receive, and the already received block it
  has to send. Schedule computations are done independently per
  processor without communication. The broadcast communication
  subgraph is an easily computable, directed, $\ceiling{\log_2
    p}$-regular circulant graph also used elsewhere. We show how
  the schedule computations can be done in optimal time and space of
  $O(\log p)$, improving significantly over previous results of
  $O(p\log^2 p)$ and $O(\log^3 p)$, respectively. The schedule computation
  and broadcast algorithms are simple to implement, but correctness and
  complexity are not obvious. The schedules are used for new
  implementations of the MPI (Message-Passing Interface) collectives
  \mpibcast, \mpiallgatherv,\mpireduce and
  \mpireducescatter. Preliminary experimental results are given.
  Carefully engineered and extensively evaluated implementations will
  be presented elsewhere.
\end{abstract}

\section{Introduction}

The broadcast problem is a fundamental problem in message-passing
computing and broadcasting a fundamental, collective operation in
programming interfaces for distributed-memory parallel computing.  In
this paper, we give a particular solution in terms of pipelining that
can be used to implement this operation (\mpibcast) as defined for the
Message-Passing Interface (MPI)~\cite{MPI-4.1} for any number of
processes, and show that our resulting, optimal schedules (in terms of
preprocessing time and space and in terms of number of communication
rounds) and circulant communication graphs can also be used to
implement the related collective communication operations
all-broadcast (\mpiallgatherv), reduction (\mpireduce) and
all-reduction (\mpireducescatterblock and \mpireducescatter). This
addresses and solves a number of problems and conjectures of recent
papers~\cite{Traff22:bcastba,Traff22:bcast,Traff23:circulant}.

The broadcasting problem considered here is the following. In a
distributed memory system with $p$ processors, a designated root
processor has $n$ indivisible blocks of data that have to be
communicated to all other processors in the system. Each processor can
in a communication operation send an already known block to some other
processor and at the same time receive a(n unknown, new) block from
some other, possibly different processor. Blocks can be sent and
received in unit time, where the time unit depends on the size of the
blocks which are assumed to all have (roughly) the same size. Since
communication of blocks takes the same time, the complexity of an
algorithm for solving the broadcast problem can be stated in terms of
the number of communication rounds
that are required for the last processor to have received all
$n$ blocks from the root.  In this fully-connected, one-ported, fully
(send-receive) bidirectional $p$ processor
system~\cite{BarNoyKipnis94,BarNoyKipnisSchieber00}, any broadcast
algorithm requires $n-1+\ceiling{\log_2 p}$ communication rounds as is
well known.  A number of algorithms reach this optimal number of
communication rounds with different communication patterns in a fully
connected
network~\cite{BarNoyKipnisSchieber00,Jia09,KwonChwa95,Traff22:bcast,Traff08:optibcast}.

The optimal communication round algorithm given
in~\cite{Traff08:optibcast} was used to implement the \mpibcast
operation for MPI~\cite{MPI-4.1}. Thus a concrete, implementable
solution was given, unfortunately with a much too high schedule
precomputation cost of $O(p\log^2 p)$ sequential steps which was
partially amortized through caching of schedules with
communicators~\cite{Traff06:mpisxcoll}. An advantage of this algorithm
compared to other solutions to the broadcast
problem~\cite{BarNoyKipnisSchieber00,Jia09,KwonChwa95} is its simple,
$\ceiling{\log_2 p}$-regular circulant graph communication pattern,
where all processors throughout the broadcasting operation operate
symmetrically. Circulant graphs with different skips are used in many
attractive algorithms for different collective communication
problems~\cite{BarNoyKipnisSchieber93,BernaschiIannelloLauria03,BruckHo93,Bruck97,Traff23:circulant}. The
circulant graph pattern combined with optimal broadcast schedules
makes it possible to implement also the the difficult, irregular
all-broadcast operation \mpiallgatherv more efficiently as was shown
recently in~\cite{Traff22:bcast}. As we observe in this paper, we can
likewise implement the reduction (\mpireduce) and the difficult,
irregular as well as regular all-reduction (\mpireducescatter,
\mpireducescatterblock) operations. These implementations answer
questions posed in~\cite{Traff23:circulant}.

In~\cite{Traff22:bcastba,Traff22:bcast}, a substantial improvement of
the schedule computation cost compared to~\cite{Traff08:optibcast} was
given, from super-linear $O(p\log^2p)$ to $O(\log^3 p)$ time steps,
thus providing a practically much more relevant algorithm. However, no
adequate correctness proof was presented.  A major challenge was to
get the schedule computation down to $O(\log p)$ time steps. This is
optimal, since at least $\ceiling{\log_2 p}$ communication rounds are
required independently of $n$ in which each processor sends and
receives different blocks.

In this paper, we constructively prove the conjecture that round
optimal send and receive schedules can be computed in $O(\log p)$
operations per processor (without any communication) by stating and
analyzing the corresponding algorithms. The new algorithms use the
same circulant graph communication pattern and give rise to the same
schedules as those constructed by previous
algorithms~\cite{Traff22:bcastba,Traff22:bcast,Traff08:optibcast} but
the algorithmic ideas are different. They are readily implementable
and of great practical relevance, as indicated by preliminary
experiments, for the full spectrum of \mpibcast, \mpiallgather,
\mpiallgatherv, \mpireduce, \mpireducescatterblock and
\mpireducescatter collectives. The paper focuses on the schedule
properties and algorithms for schedule computation. More careful,
tuned implementation and experimental evaluation of the collective
operations will be given elsewhere. Nevertheless, we think the
algorithms presented here can become algorithms of choice for MPI
libraries and implementers.

\section{Broadcast(ing with) Schedules and Circulant Graphs}
\label{sec:broadcasts}

Our problem is to distribute (broadcast) $n$ indivisible, disjoint
blocks of data from a single, designated \emph{root} processor to all
other processors in a $p$-processor system with processors $r,0\leq
r<p$ that each communicate with certain other processors by
simultaneously sending and receiving data blocks. The $n$ blocks for
instance constitute a buffer of $m$ data units, $m\geq n$, which can
be broadcast by broadcasting the $n$ blocks, each of size at most
$\ceiling{m/n}$ units. Choosing a best $n$ for a given $m$ is a system
dependent tuning problem that is not addressed at all in this paper.

We first show how broadcast from a designated root, $r=0$ (without
loss of generality) can be done with a regular, symmetric
communication pattern and explicit send and receive schedules that in
constant time determine for each communication operation by each
processor which block is received and which block is sent. We use this
algorithm to formulate the correctness conditions on possible send and
receive schedules. Based on this, we then in the next section show how
correct schedules can be computed fast in $O(\log p)$ time per
processor, independently of all other processors and with no
communication.

In all of the following, we let $p$ denote the number of processors
and take $q=\ceiling{\log_2 p}$.

\begin{algorithm}
  \caption{The $n$-block broadcast algorithm for processor $r,0\leq
    r<p$ of data blocks in array \texttt{buf}. Round $x$ numbers
    the first round where actual communication takes place. Blocks
    smaller than $0$ are neither sent nor received, and for blocks
    larger than $n-1$, block $n-1$ is sent and received instead. Also
    blocks to the root processor are not sent. This is assumed to be
    taken care of by the bidirectional $\bidirec{}{}$ communication
    operations.}
  \label{alg:broadcast}
  \begin{algorithmic}
    \Procedure{Broadcast}{$\mathtt{buf}[n]$}
    \State$\Call{recvschedule}{r,\recvblock{}}$
    \State$\Call{sendschedule}{r,\sendblock{}}$
    \State $x\gets (q-(n-1)\bmod q)\bmod q$ 
    \For{$i=0,1,\ldots,q-1$} 
    \State$\mathtt{recvblock}[i]\gets\recvblock{i}-x$
    \State$\mathtt{sendblock}[i]\gets\sendblock{i}-x$
    \If{$i<x$} \Comment Virtual rounds before $x$ already done
    \State$\mathtt{recvblock}[i]\gets\recvblock{i}+q$
    \State$\mathtt{sendblock}[i]\gets\sendblock{i}+q$
    \EndIf
    \EndFor
    \For{$i=x,x+1,\ldots,(n+q-1+x)-1$}
    \State $k\gets i\bmod q$
    \State $t\gets (r+\skips{k})\bmod p$
    \Comment to- and from-processors
    \State $f\gets (r-\skips{k}+p)\bmod p$
    \State $\bidirec{\mathtt{buf}[\sendblock{k}],t}{\mathtt{buf}[\recvblock{k}],f}$
    \State $\sendblock{k}\gets \sendblock{k}+q$
    \State $\recvblock{k}\gets \recvblock{k}+q$
    \EndFor
    \EndProcedure
\end{algorithmic}
\end{algorithm}

The generic $n$-block broadcast algorithm is shown as
Algorithm~\ref{alg:broadcast}. It proceeds in rounds and is symmetric
in the sense that all processors follow the same regular, circulant
graph communication pattern and do the same communication operations
in each round.  For the rooted, asymmetric broadcast operation, this
is perhaps surprising. Rounds start from some round $x$ (to be
explained shortly) and end at round $n-1+q+x-1$ for a total of the
required $n-1+q$ communication rounds.  For round $i,x\leq i<n-1+q+x$,
we take $k=i\bmod q$, such that always $0\leq k<q$. In round $i$, each
processor $r, 0\leq r<p$ simultaneously sends a block to a
\emph{to-processor} $t^k_r=(r+\skips{k})\bmod p$ and receives a
different block from a \emph{from-processor} $f^k_r
=(r-\skips{k}+p)\bmod p$, determined by a skip per round $\skips{k},
0\leq k<q$. These skips are the same for all processors and define the
circulant graph.  The blocks that are sent and received are numbered
consecutively from $0$ to $n-1$ and stored in a \texttt{buf} array
indexed by the block number. The block that processor $r,0\leq r<p$
sends in round $i$ is determined by a send schedule array
$\sendblock{k}_r$ and likewise the block that processor $r$ will
receive in round $i$ by a receive schedule array
$\recvblock{k}_r$. Since the blocks are thus fully determinate, no
block indices or other meta-data information needs to be
communicated. The $\sendblock{}_r$ and $\recvblock{}_r$ arrays are
computed such that blocks are effectively broadcast from root
processor $r=0$ that initially has all $n$ blocks, and such that all
$n$ blocks are received and sent further on by all the other
processors. The starting round $x$ is chosen such that $(n-1+q+x)\bmod
q=(n-1+x)\bmod q=0$ and after this last round, which is a multiple of
$q$, all processors will have received all $n$ blocks. The assumption
that processor $r=0$ is the root can be made without loss of
generality: Should some other processor $r'\neq 0$ be root, the processors
are simply renumbered by subtracting $r'$ (modulo $p$) from the
processor indices.

To make Algorithm~\ref{alg:broadcast} efficient and correct, no block
is ever sent back to the root (which has all blocks in the first
place), so there is no send operation if $t^k_r=0$ for processor $r$
in round $k$. Non-existent, negatively indexed blocks are neither sent
nor received, so if $\sendblock{k}_r<0$ or $\recvblock{k}_r<0$ for some
$k$, the corresponding send or receive communication is simply
ignored. For block indices larger than the last block $n-1$, block
$n-1$ is instead sent or received. These cases are assumed to be
handled by the concurrent send- and receive operations denoted by
$\bidirec{}{}$. The receive and send block schedules $\recvblock{}_r$
and $\sendblock{}_r$ are computed by the calls to
\textsc{recvschedule()} and \textsc{sendschedule()} to be derived in
Section~\ref{sec:recvschedule} and Section~\ref{sec:sendschedule}.
Receive and send schedules for $p=17$ (a non-power of $2$) with $q=5$ is
shown in Table~\ref{tab:p17}, which can be used to trace the communication
when broadcasting with Algorithm~\ref{alg:broadcast}.

\begin{table}
  \caption{Receive and send schedule for a non-power-of-two number of
    processors, $p=17$, $q=\ceiling{\log_2 p}=5$.  The table shows for
    each processor $r,0\leq r<p$ the baseblock $b$ and the $\recvblock{k}_r$ and
    $\sendblock{k}_r$ schedules for $k=0,1,2,3,4$. The vertical bars correspond
    to $\skips{k}$}
  \label{tab:p17}
  \begin{center}
  \begin{tabular}{cr|r|r|rr|rrrr|rrrrrrrrrr|}
$r$: &  0 & 1 & 2 & 3 & 4 & 5 & 6 & 7 & 8 & 9 & 10 & 11 & 12 & 13 & 14 & 15 & 16 \\
$b$: &5 & 0 & 1 & 2 & 0 & 3 & 0 & 1 & 2 & 4 & 0 & 1 & 2 & 0 & 3 & 0 & 1 \\
\toprule
$\recvblock{0}$: & -4 & 0 & -5 & -4 & -3 & -5 & -2 & -5 & -4 & -3 & -1 & -5 & -4 & -3 & -5 & -2 & -5 \\
$\recvblock{1}$: & -5 & -4 & 1 & -5 & -4 & -3 & -3 & -2 & -5 & -4 & -3 & -1 & -5 & -4 & -3 & -3 & -2 \\
$\recvblock{2}$: & -2 & -2 & -2 & 2 & 0 & -4 & -4 & -3 & -2 & -2 & -4 & -3 & -1 & -1 & -4 & -4 & -3 \\
$\recvblock{3}$: & -1 & -3 & -3 & -2 & -2 & 3 & 0 & 1 & 2 & -5 & -2 & -2 & -2 & -2 & -1 & -1 & -1 \\
$\recvblock{4}$: & -3 & -1 & -1 & -1 & -1 & -1 & -1 & -1 & -1 & 4 & 0 & 1 & 2 & 0 & 3 & 0 & 1 \\
\midrule
$\sendblock{0}$: & 0 & -5 & -4 & -3 & -5 & -2 & -5 & -4 & -3 & -1 & -5 & -4 & -3 & -5 & -2 & -5 & -4 \\
$\sendblock{1}$: & 1 & -5 & -4 & -3 & -3 & -2 & -5 & -4 & -3 & -1 & -5 & -4 & -3 & -3 & -2 & -5 & -4 \\
$\sendblock{2}$: & 2 & 0 & -4 & -4 & -3 & -2 & -2 & -4 & -3 & -1 & -1 & -4 & -4 & -3 & -2 & -2 & -2 \\
$\sendblock{3}$: & 3 & 0 & 1 & 2 & -5 & -2 & -2 & -2 & -2 & -1 & -1 & -1 & -1 & -3 & -3 & -2 & -2 \\
$\sendblock{4}$: & 4 & 0 & 1 & 2 & 0 & 3 & 0 & 1 & -3 & -1 & -1 & -1 & -1 & -1 & -1 & -1 & -1 \\
\bottomrule
  \end{tabular}
  \end{center}
\end{table}

For the algorithm to be correct in the sense of broadcasting all
blocks from root processor $r=0$ to all other processors, the following
conditions must hold:

\begin{enumerate}
  \item
    The block that is received in round $i$ with $k=i\bmod q$ by some
    processor $r$ must be the block that is sent by the from-processor
    $f_r^k$, $\recvblock{k}_r = \sendblock{k}_{f_r^k}$.
  \item
    The block that processor $r$ sends in round $i$ with $k=i\bmod q$
    must be the block that the to-processor $t_r^k$ will receive,
    $\sendblock{k}_r = \recvblock{k}_{t_r^k}$.
  \item
    \label{cor:qblocks}
    Over any $q$ successive rounds, each processor must receive $q$
    different blocks. More concretely,
    $\bigcup_{k=0}^{q-1}\{\recvblock{k}_r\}
    =(\{-1,-2,\ldots,-q\}\setminus\{b_r-q\})\cup\{b_r\}$ where $b_r, 0\leq
    b_r<q$ is a non-negative, real block received by 
    processor $r$ in one of the first $q$ rounds. This block $b_r$ is called
    the \emph{baseblock} for processor $r$.
  \item
    \label{cor:recvbefore}
    The block that a processor sends in round $i$ with $k=i\bmod q$
    must be a block that has been received in some previous round, so
    either $\sendblock{k}_r=\recvblock{j}_r$ for some $j, 0\leq j<k$, or
    $\sendblock{k}_r=b_r-q$ where $b_r\geq 0$ is the \emph{baseblock} for
    processor $r$.
\end{enumerate}

The last correctness condition implies that $\sendblock{0}_r=b_r-q$ for
each processor. With receive and send schedules fulfilling these four
conditions, it is easy to see that Algorithm~\ref{alg:broadcast} correctly
broadcasts the $n$ blocks over the $p$ processors.

\begin{theorem}
  \label{thm:scheduleuse}
  Let $K,K>0$ be a number of communication phases each consisting of
  $q$ communication rounds for a total of $Kq$ rounds. Assume that in
  each round $i, 0\leq i<Kq$, each processor $r,0\leq r<p$ receives a
  block $\recvblock{i\bmod q}_r+\floor{i/q}q$
  and sends a block
  $\sendblock{i\bmod q}_r+\floor{i/q}q$.
  By the end of the $Kq$ rounds, processor $r$ will have received all blocks
  $\{0,1,\ldots,(K-1)q-1\}\cup \{b_r+(K-1)q\}$ where $b_r$ is the first
  (non-negative) block received by processor $r$.
\end{theorem}
\begin{proof}
  The proof is by induction on the number of phases. For $K=1$, there
  are $q$ rounds $i=0,1,\ldots,q-1$ over which each processor will
  receive its non-negative baseblock $b$; all other receive blocks are
  negative (Correctness Condition~\ref{cor:qblocks}). For $K>1$, in
  the last phase $K-1$, each processor will receive the blocks
  $(\{(K-2)q,(K-2)q+1,\ldots,(K-2)q+q-1\}\setminus\{b+(K-2)q\})\cup\{b+(K-1)q\}$
  since the set $\bigcup_{k=0}^{q-1}\{\recvblock{k}\}$ contains $q$
  different block indices, one of which is positive. The block
  $b+(K-2)q$ has been received in phase $K-2$ by the induction
  hypothesis, in its place block $b+(K-1)q$ is received. Therefore, at
  the end of phase $K-1$ using the induction hypothesis, blocks
  $\{0,1,\ldots,(K-2)q-1\}\cup \{(K-2)q,(K-2)q+1,\ldots, (K-2)q+q-1\}
  = \{0,1,\ldots, (K-1)q-1\}$ plus the block $b+(K-1)q$ have been
  received, as claimed. By Correctness Condition~\ref{cor:recvbefore},
  no block is sent that has not been received in a previous round or
  phase.
\end{proof}

In order to broadcast a given number of blocks $n$ in the optimal
number of rounds $n-1+q$, we use the smallest number of phases $K$
such that $Kq\geq n-1+q$, and introduce a number of dummy blocks
$x=Kq-(n-1+q)$ that do not have to be broadcast. In the $K$ phases,
all processors will have received the $n+x-1$ blocks
$0,1,\ldots,n+x-2$ plus one larger block. We perform $x$ initial,
virtual rounds with no communication to handle the $x$ dummy blocks,
and broadcast the actual, positive blocks in the following $n-1+q$
rounds. This is handled by subtracting $x$ from all computed block
indices. Recall that negative blocks are neither received nor sent by the
algorithm.  Blocks with index larger than $n-1$ in the last phase are
capped to $n-1$.

The symmetric communication pattern where each processor (node) $r$
has incoming (receive) edges $(f_r^k,r)$ and outgoing (send) edges
$(r,t^k_r)$ is a \emph{circulant graph} with skips (jumps)
$\skips{k}, k=0,1,\ldots, q-1$. The skips are computed in
Section~\ref{sec:circulant}.

By Theorem~\ref{thm:scheduleuse}, we make the following observations.

\begin{observation}
  \begin{enumerate}
  \item
    Algorithm~\ref{alg:broadcast} can be used to implement the
    \mpibcast operation for any number of processes and any problem
    size $m$. The data of size $m$ units is divided into $n$ blocks to
    give the shortest completion time (empirical tuning problem). With
    only $n=1$ block, the algorithm takes $q=\ceiling{\log_2 p}$
    communication rounds and is similar to a binomial tree.
  \item
    Since the communication pattern is symmetric and all processes do
    the structurally same communication operations in each round,
    simultaneous broadcast from all processes as roots can be done at
    the same time. For this, each processes computes send and receive
    schedules with each other process as root. In each send and receive
    operations, $p-1$ blocks corresponding to the $p-1$ other roots
    are concatenated and treated as one unit. Thus, the all-broadcast
    (all-to-all broadcast, allgather) problem where each process has
    $n$ blocks (possibly of different sizes) to broadcast can be
    solved with $n-1+q$ communication rounds. The algorithm is shown
    in the appendix.  By this, the difficult, irregular \mpiallgatherv
    operation can be implemented efficiently.
  \item
    By reversing the direction of all communication operations,
    working downwards from round $(n-1+q+x)-1$ to round $x$ and
    applying some given binary, associative and commutative reduction
    operator on each received block, the schedule gives a
    round-optimal solution to the reduction problem for commutative
    operators.  Each non-root process will send each partial result
    block exactly once, and the root will receive and reduce partial
    results for all blocks. This observation can be used to implement
    \mpireduce in the optimal number of rounds for any problem size
    $m$ and any number of processes.
  \item
    Simultaneous reduction to all processes as roots can likewise be
    done, giving round-optimal solutions for the all-reduction operations
    \mpireducescatterblock and the difficult, irregular
    \mpireducescatter. Each partial result block for all processes is sent
    and received once for a total volume of $p-1$ blocks sent and
    $p-1$ blocks received. For $n=1$ block, we believe that this is the
    first such algorithm with logarithmic number of rounds for any number
    of processors.  Previous algorithms use either a linear
    ring~\cite{ChanHeimlichPurkayasthavandeGeijn07,Traff10:largeallgat}
    or have almost twice the communication volume~\cite{Traff23:circulant}
    for certain number of processes $p$.
  \end{enumerate}
\end{observation}

\begin{table}
  \caption{Receive and send schedules for $p=9$ processors with $q=4$.}
  \label{tab:p9}
  \begin{center}
  \begin{tabular}{cr|r|r|rr|rrrr|}
$r$ & 0 & 1 & 2 & 3 & 4 & 5 & 6 & 7 & 8 \\
$b$ & 4 & 0 & 1 & 2 & 0 & 3 & 0 & 1 & 2 \\
\toprule
$\recvblock{0}$: & -2 & 0 & -4 & -3 & -2 & -4 & -1 & -4 & -3 \\
$\recvblock{1}$: & -3 & -2 & 1 & -4 & -3 & -2 & -2 & -1 & -4 \\
$\recvblock{2}$: & -1 & -3 & -2 & 2 & 0 & -3 & -3 & -2 & -1 \\
$\recvblock{3}$: & -4 & -1 & -1 & -1 & -1 & 3 & 0 & 1 & 2 \\
\midrule
$\sendblock{0}$: & 0 & -4 & -3 & -2 & -4 & -1 & -4 & -3 & -2 \\
$\sendblock{1}$: & 1 & -4 & -3 & -2 & -2 & -1 & -4 & -3 & -2 \\
$\sendblock{2}$: & 2 & 0 & -3 & -3 & -2 & -1 & -1 & -3 & -2 \\
$\sendblock{3}$: & 3 & 0 & 1 & 2 & -4 & -1 & -1 & -1 & -1 \\
\bottomrule
  \end{tabular}
  \end{center}
\end{table}

\begin{table}
  \caption{Receive and send schedules for $p=18$ processors with $q=5$.}
  \label{tab:p18}
  \begin{center}
  \begin{tabular}{cr|r|r|rr|rrrr|rrrrrrrrr|}
$r$ & 0 & 1 & 2 & 3 & 4 & 5 & 6 & 7 & 8 & 9 & 10 & 11 & 12 & 13 & 14 & 15 & 16 & 17 \\
$b$ & 5 & 0 & 1 & 2 & 0 & 3 & 0 & 1 & 2 & 4 & 0 & 1 & 2 & 0 & 3 & 0 & 1 & 2 \\
\toprule
$\recvblock{0}$: & -3 & 0 & -5 & -4 & -3 & -5 & -2 & -5 & -4 & -3 & -1 & -5 & -4 & -3 & -5 & -2 & -5 & -4 \\
$\recvblock{1}$: & -4 & -3 & 1 & -5 & -4 & -3 & -3 & -2 & -5 & -4 & -3 & -1 & -5 & -4 & -3 & -3 & -2 & -5 \\
$\recvblock{2}$: & -2 & -4 & -3 & 2 & 0 & -4 & -4 & -3 & -2 & -2 & -4 & -3 & -1 & -1 & -4 & -4 & -3 & -2 \\
$\recvblock{3}$: & -5 & -2 & -2 & -2 & -2 & 3 & 0 & 1 & 2 & -5 & -2 & -2 & -2 & -2 & -1 & -1 & -1 & -1 \\
$\recvblock{4}$: & -1 & -1 & -1 & -1 & -1 & -1 & -1 & -1 & -1 & 4 & 0 & 1 & 2 & 0 & 3 & 0 & 1 & 2 \\
\midrule
$\sendblock{0}$: & 0 & -5 & -4 & -3 & -5 & -2 & -5 & -4 & -3 & -1 & -5 & -4 & -3 & -5 & -2 & -5 & -4 & -3 \\
$\sendblock{1}$: & 1 & -5 & -4 & -3 & -3 & -2 & -5 & -4 & -3 & -1 & -5 & -4 & -3 & -3 & -2 & -5 & -4 & -3 \\
$\sendblock{2}$: & 2 & 0 & -4 & -4 & -3 & -2 & -2 & -4 & -3 & -1 & -1 & -4 & -4 & -3 & -2 & -2 & -4 & -3 \\
$\sendblock{3}$: & 3 & 0 & 1 & 2 & -5 & -2 & -2 & -2 & -2 & -1 & -1 & -1 & -1 & -5 & -2 & -2 & -2 & -2 \\
$\sendblock{4}$: & 4 & 0 & 1 & 2 & 0 & 3 & 0 & 1 & 2 & -1 & -1 & -1 & -1 & -1 & -1 & -1 & -1 & -1 \\
\bottomrule
  \end{tabular}
  \end{center}
\end{table}

We can make an important, structural observation on receive
schedules that fulfill the four correctness conditions which helps to
understand the intuition of the broadcast algorithm described as
Algorithm~\ref{alg:broadcast} and the construction of receive and
send schedules to be developed in the next sections.

\begin{observation}
  \label{obs:double}
  Let a correct receive schedule for $p$ processors be given. A
  correct receive schedule for $2p$ processors can be constructed by
  the following modifications. First, for each processor $r, p\leq
  r<2p$, copy the receive schedule for processor $r-p$, such that
  $\recvblock{k}_r=\recvblock{k}_{r-p}$ for each round $k, 0\leq k<q$.
  Thus, the block that a processor would have received in one of the
  first $q$ rounds in the $p$ processor schedule is the same as it
  will now receive in the $2p$ processor schedule. Since $q$ is one
  larger in the $2p$ than in the $p$ processor schedule, subtract $1$
  from all the negative blocks in the receive schedules. Since $p$ is
  doubled, set $\skips{q+1}=2p$, all other $\skips{k}, 0\leq k\leq $
  will stay. The receive blocks $\recvblock{q}_r$ for all processors
  $r,0\leq r<2p$ are determined as follows.  For processor $r=p$, this
  shall be the new baseblock $q$ that has not been used so far, so
  $\recvblock{q}_p=q$. For all other large processors $r,p<r<2p$, 
  find the positive baseblock $b_r$ in $\recvblock{k}_r$, replace this
  with $-1$ (which will correspond to receiving baseblock $q$ received
  previously by processor $r$) and set $\recvblock{q}_r=b_r$.  For
  processors $r, 0\leq r<p$, we only need to set $\recvblock{q}_r=-1$.
\end{observation}

Table~\ref{tab:p9} and Table~\ref{tab:p18} show receive and send
schedules for $p=9$ and $p=18$ processors. The doubling of $p$ from
$9$ to $18$ gives rise to the shown schedule by following
Observation~\ref{obs:double} as can readily be verified. If the
receive schedule for $p=9$ is correct, so is obviously and by
construction the schedule for $p=18$. The observation shows that
receive schedules where $p=2^q$ is a power of two indeed exist and can
readily be computed. Extending a $p$ processor receive schedule to a
$2p-1$ processor receive schedule is not as obvious. Had it been, it
would be easy to construct $p$ processor receive schedules for any
$p$.

\subsection{The communication pattern}
\label{sec:circulant}

\begin{algorithm}
  \caption{Computing the skips (jumps) for a $p$-processor circulant
    graph ($q=\ceiling{\log_2 p}$).}
  \label{alg:circulants}
\begin{algorithmic}
  \State $k\gets q$
  \State $\skips{k}\gets p$
  \While{$k>0$}
  \State $\skips{k-1}\gets \skips{k}-\skips{k}/2$
  \State $k\gets k-1$   \Comment Now $\skips{k}=\ceiling{\skips{k+1}/2}$
  \EndWhile
\end{algorithmic}
\end{algorithm}

The skips for the circulant graph pattern are computed by repeated
halving with rounding up of $p$ and shown as
Algorithm~\ref{alg:circulants}.  This obviously takes
$q=\ceiling{\log_2 p}$ iterations. We observe that for any $p,p>1$ it
is always the case that $\skips{0}=1$ and $\skips{1}=2$.  For
convenience, we take $\skips{q}=p$.  Thus, in the regular, circulant
graph, each processor has $q$ incoming and $q$ outgoing neighbor
processors. We make some simple observations needed for the following
computations of the send and receive schedules.

\begin{observation}
  \label{obs:skips}
  For each $k, 0\leq k<q$ it holds that
  $\skips{k+1}\leq\skips{k}+\skips{k}\leq\skips{k+1}+1$
\end{observation}
This follows directly from the halving scheme of
Algorithm~\ref{alg:circulants}.

\begin{lemma}
  \label{lem:bounds}
  For any $k, k<q$ it holds that
  \begin{displaymath}
    \skips{k+1}-1\leq \sum_{i=0}^k\skips{i}<\skips{k+1}+k
  \end{displaymath}
\end{lemma}

\begin{proof}
  The prof is by induction on $k$. For $k=0$, the claim holds since
  $\skips{0}=1$ and $\skips{1}=2$. Assume the claim holds for $k$. By
  the induction hypothesis and Observation~\ref{obs:skips}, we get
  \begin{displaymath}
    \sum_{i=0}^{k+1}\skips{i} = \skips{k+1}+\sum_{i=0}^{k}\skips{i} <
    \skips{k+1}+\skips{k+1}+k\leq \skips{k+2}+k+1
  \end{displaymath}
    and
  \begin{displaymath}
    \sum_{i=0}^{k+1}\skips{i} = \skips{k+1}+\sum_{i=0}^{k}\skips{i} \geq
    \skips{k+1}+\skips{k+1}-1\geq \skips{k+2}-1 \quad .
  \end{displaymath}  
\end{proof}

\begin{observation}
  \label{obs:tunnel}
  For any $p$, there are at most two $k, k>1$ such that
  $\skips{k-2}+\skips{k-1}=\skips{k}$.
\end{observation}
For $\skips{2}=3$, Algorithm~\ref{alg:circulants} gives $\skips{1}=2$
and $\skips{0}=1$ for which the observation holds. For $\skips{3}=5$,
Algorithm~\ref{alg:circulants} gives $\skips{2}=3$ and $\skips{1}=2$
for which the observation holds. Any $p$ for which $\skips{2}=3$ or
$\skips{3}=5$ will have this property, and none other. We see for
$p\geq 3$ that
$\skips{2}\geq 3$, for $p>4$ that
$\skips{3}\geq 5$ and for $p>7$ that $\skips{4}\geq 9$, and therefore
$\skips{k-2}+\skips{k-1}<\skips{k}$ for $k>3$ for all $p$.

\begin{observation}
  \label{obs:cut}
  For some $p$ and $k>0$, there is an $r, r<\skips{k}$ with
  $r+\skips{k}=\skips{k+1}$.
\end{observation}
If $\skips{k+1}$ is odd, $r=\skips{k+1}-\skips{k}$ fulfills the
observation.


\begin{lemma}
  \label{lem:canonical}
  For any $r,0\leq r <p$ there is a (possibly empty) sequence
  $[e_0,e_1,\ldots,e_{j-1}]$ of $j, j<q$, different skip indices such
  that $r=\sum_{i=0}^{j-1} \skips{e_i}$.
\end{lemma}
We call a (possibly empty) sequence $[e_0,e_1,\ldots,e_{j-1}]$ for
which $r=\sum_{i=0}^{j-1} \skips{e_i}$ and where
$e_0<e_1<\ldots<e_{j-1}$ a \emph{skip sequence} for $r$.
A non-empty skip sequence $[e_0,e_1,\ldots,e_{j-1}]$ for $r$ defines a
path from processor $0$ to processor $r>0$.
We will use the terms skip sequence and path interchangeably.

\begin{proof}
  The proof is by induction on $k$. When $r=0$ the claim holds for the
  empty sequence. Assuming the claim holds for any $r, 0\leq
  r<\skips{k}$, we show that it holds for $0\leq r<\skips{k+1}$. If
  already $0\leq r<\skips{k}$ the claim holds by assumption.  If
  $\skips{k}\leq r <\skips{k+1}$, then $0\leq
  r-\skips{k}<\skips{k+1}-\skips{k}\leq\skips{k}$ by
  Observation~\ref{obs:skips}. By the induction hypothesis, there is a
  sequence of different skips not including $\skips{k}$ and summing to
  $r-\skips{k}$, and $\skips{k}$ can be appended to this sequence to
  sum to $r$.
\end{proof}

The lemma indicates how to recursively compute in $O(q)$ steps a
specific, \emph{canonical skip sequence} for any $r, 0\leq r<p$. By
Observation~\ref{obs:tunnel} and Observation~\ref{obs:cut}, for some
$p$ there may be more than one skip sequence for some $r$; the
decomposition of $r$ into a sum of different skips is not unique for
all $p$ (actually, the decomposition is unique only when $p$ is a
power of $2$).  A canonical skip sequence will contain $\skips{k}$ and
not $\skips{k-2}$ and $\skips{k-1}$ if
$\skips{k-2}+\skips{k-1}=\skips{k}$ (Observation~\ref{obs:tunnel}),
and $\skips{k+1}$ instead of $\skips{k}$ if $r+\skips{k}=\skips{k+1}$
(Observation~\ref{obs:cut}).

\begin{algorithm}
  \caption{Finding the baseblock for processor $r,0\leq r<p$.}
  \label{alg:baseblock}
\begin{algorithmic}
  \Function{baseblock}{$r$}
  \State $k, r'\gets q, 0$
  \Repeat
  \State $k\gets k-1$
  \If{$r'+\skips{k}=r$} \Return{$k$}
  \ElsIf{$r'+\skips{k}<r$} $r'\gets r'+\skips{k}$
  \EndIf
  \Until{$k=0$}
  \State \Return{$q$} \Comment Only processor $r=0$ will return $q$ as baseblock
  \EndFunction
\end{algorithmic}
\end{algorithm}

The canonical skip sequence for an $r,0\leq r<p$ is implicitly
computed iteratively by the \textsc{baseblock()} function of
Algorithm~\ref{alg:baseblock} which explicitly returns the first
(smallest) skip index in the canonical skip sequence.  This index will
be the \emph{baseblock} $b_r$ for $r$. Baseblocks are essential for
the broadcast schedules. For convenience, we define $q$ to be the
baseblock of $r=0$ for which the skip sequence is otherwise empty; for
other $r>0$ the baseblock $b_r$ satisfies $0\leq b_r<q$ and is a legal
skip index.

If we let the root processor $r=0$ send out blocks one after the other
(that is, $\sendblock{k}_0=k$ for $k=0,1,\ldots,q-1$) to processors
$\skips{k}$, the canonical skip sequence to any other processor $r,
r>0$ gives a path through which the baseblock $b_r$ can be sent from
the root to processor $r$. Processor $r$ will receive its baseblock in
communication round $e$ where $e$ is the last, largest skip index in
the canonical skip sequence for $r$, therefore we can set
$\recvblock{e}_r=b_r$.

The distribution of baseblocks for processors $r=0,1,2,\ldots,p-1$ follows
a somewhat, and for $p=2^q$ particularly regular pattern. The next lemma
shows that for any consecutive sequence of processors, there are enough
corresponding baseblocks that any processor can in $k$ communication rounds
receive $k$ different baseblocks when using the circulant graph communication
pattern.

\begin{lemma}
  \label{lem:kplus1}
  Any consecutive sequence of baseblocks $b_r,b_{r+1},\ldots
  b_{r+\skips{k}-1}$ of length $\skips{k}$ for processors
  $r,r+1,\ldots, r+\skips{k}-1$ has at least $k+1$ different
  baseblocks.
\end{lemma}

\begin{proof}
  The proof is by induction on $k$. For this we list the baseblocks in
  sequence for all processors $r,0\leq r<p$ which can be done in
  $O(p)$ steps as follows. At first, only the root processor $r=0$ has
  baseblock $0$. Assume that the baseblocks for processors
  $0,1,\ldots,\skips{k}-1$ have been listed.  Append this list to
  itself but possibly remove the last element, giving now a list of
  $\skips{k+1}$ baseblocks. Increment the baseblock for processor
  $r=0$ to $k+1$. Example for $\skips{0}=1$, $\skips{1}=2$,
  $\skips{2}=3$, $\skips{3}=6$, $\skips{4}=11$: $0\rightarrow 10
  \rightarrow 201\rightarrow 301201\rightarrow 40120130120$. As can be
  seen, each sequence of $\skips{k}$ baseblocks, starting from $r=0$
  has exactly $k+1$ different baseblocks. Furthermore, any sequence
  starting from $r=\skips{k}$ has likewise $k+1$ different baseblocks,
  even when one block is missing from the end when
  $\skips{k}+\skips{k}=\skips{k+1}+1$.  From the construction it can
  also be seen that in any sequence of baseblocks for a consecutive
  sequence of processors, there is a unique, largest baseblock
  appearing once. Now consider a sequence of $\skips{k}$
  baseblocks. Find the unique, largest baseblock corresponding to some
  processor $r$. Within the $\skips{k}$ length sequence, there is a
  consecutive, shorter sequence of length $\skips{k-1}+1$ either
  starting or ending at $r$. By the induction hypothesis, the sequence
  of length $\skips{k-1}$ has at least $k$ different baseblocks and
  does not include the baseblock of processor $r$. In total, there are
  therefore at least $k+1$ different baseblocks.
\end{proof}

The way the $p$ baseblocks are listed correspond to what
Algorithm~\ref{alg:baseblock} is doing for each individual processor
$r$. The proof shows that all baseblocks can be listed more
efficiently, in linear time, than by application of
\textsc{baseblock} for each $r$. This is relevant when the schedule
computations are used for all-broadcast and all-reduction.

\subsection{The receive schedule}
\label{sec:recvschedule}

\begin{algorithm}
  \caption{Computing receive blocks for processor $r, p\leq r<2p$ by
    depth-first search with removal of accepted skip indices.}
  \label{alg:dfsblocks}
  \begin{algorithmic}
    \Function{allblocks}{$r,r',s,e,k,\recvblock{q}$}
    \While{$e\neq -1$}
    \If{$r'+\skips{e}\leq r-\skips{k}\wedge r'+\skips{e}<s$}
    \If{$r'+\skips{e}\leq r-\skips{k+1}$}
    \State $k\gets \Call{allblocks}{r,r'+\skips{e},s,e,k,\recvblock{}}$
    \EndIf
    \If{$r'>r-\skips{k+1}$} \Return $k$
    \EndIf
    \State $s\gets r'+\skips{e}$
    \Comment Canonical skip sequence found, keep it in $s$
    \State $\recvblock{k},k\gets e,k+1$ \Comment Accept $e$, next $k$
    \State $\nextlink{\prevlink{e}},\prevlink{\nextlink{e}}\gets \nextlink{e},\prevlink{e}$ \Comment Unlink $e$ 
    \EndIf
    \State $e\gets \nextlink{e}$
    \EndWhile
    \State \Return $k$
    \EndFunction
  \end{algorithmic}
\end{algorithm}

We now show how to compute the receive schedule $\recvblock{k},
\allowbreak k=0,\ldots q-1$ for any processor $r,0\leq r<p$ in $O(q)$
operations. More precisely, we compute for any given $r$, the $q$
blocks that $r$ will receive in the $q$ successive communication
rounds $k=0,1,\ldots,q-1$ when processor $0$ is the root processor.
The basis for the receive schedule computation is to find $q$
canonical skip sequences to different intermediate processors $r'_k$
for $k=0,1,\ldots,q-1$ with $r'_k\leq r-\skips{k}$ and
$r_0>r'_1>\ldots>r'_{q-1}>0$. That all $r'_k$ are positive will be
ensured by considering processor $p+r$ instead of $r$ as justified by
Observation~\ref{obs:double} since processors $r$ and $r+p$ have
essentially the same receive schedules. Hereby, modulo computations in
the circulant graphs are avoided.  The $k$th skip sequence to $r'_k$
is not allowed to contain the baseblock of any of the previously found
skip sequences. Therefore, the skip sequences may not fulfill exactly
$r'_k=r-\skips{k}$, but it must be ensured that $r'_k\leq r-\skips{k}$ since
$r-\skips{k}$ is the processor from which processor $r$ will receive a
new block in round $k$ in the broadcast algorithm listed in
Algorithm~\ref{alg:broadcast}. As will be seen, the skip sequences
furthermore fulfil that $r-\skips{k+1}\leq r'_k\leq r-\skips{k}$.
These properties will ensure Correctness Condition~\ref{cor:qblocks}
on the receive schedule that over $q$ rounds, a processor will receive
$q$ different blocks, and since $r'_k$ is not too far away from
$r-\skips{k}$, this block can indeed have been delivered to processor
$r-\skips{k}$ prior to communication round $k$ where $r$ is receiving
from $r-\skips{k}$ (Lemma~\ref{lem:bounds}).

For $q$ skips, there are $2^q\geq p$ canonical skip sequences, so
exploring them all will give a linear time (or worse) and not an
$O(\log p)$ time algorithm. Also, using Algorithm~\ref{alg:baseblock}
$q$ times in succession will give an $O(q^2)$ time
algorithm. Observation~\ref{obs:double} is only for even $p$ and will,
naively implemented, also give an $O(\log^2 p)$ algorithm. We instead
do a greedy search through the canonical skip sequences by the special
backtracking algorithm shown as Algorithm~\ref{alg:dfsblocks}.  The
complexity of the computation is reduced to $O(q)$ steps by removing
the smallest (last found) skip index, corresponding to the baseblock
for $r'_k$, each time a good canonical skip sequence has been found,
and then backtracking only enough to find the next, $k+1$st skip
sequence to $r'_{k+1}$.  The smallest (last found) skip index of the
$k$th canonical skip sequence will be taken as $\recvblock{k}$.  The
\textsc{allblocks()} function of Algorithm~\ref{alg:dfsblocks} stores
the (remaining) skip indices in a doubly linked list in decreasing
order. Thus, for a skip index $e$, $\nextlink{e}$ is the next,
smaller, remaining skip index. This list is used to try the skips in
decreasing order as in Algorithm~\ref{alg:baseblock}. Once an $r'_k$
is found, the corresponding skip index $e$ (baseblock) is removed by
linking it out of the doubly linked list in $O(1)$ time as also shown
in Algorithm~\ref{alg:dfsblocks}.

In the \textsc{allblocks()} function, $r'$ is the largest intermediate
processor found so far. Each recursive call looks for the largest,
remaining skip index $e$ for which $r'+\skips{e}\leq r-\skips{k}$
where furthermore $r'+\skips{e}$ is different from $r'_{k-1}$, causing
a traversal of some of the remaining skip indices in the list. If
furthermore also $r'+\skips{e}\leq r-\skips{k+1}$, a recursive call is
done on $r'+\skips{e}$ to find an intermediate processor even closer to
$r-\skips{k}$.  Upon return, processor $r'+\skips{e}$ is taken as
$r'_k$ if still $r'\leq r-\skips{k+1}$ and $e$ will be the baseblock stored
as $\recvblock{k}$.

\begin{algorithm}
  \caption{Computing the receive schedule for processor $r,0\leq r<p$.}
  \label{alg:recvschedule}
  \begin{algorithmic}
    \Procedure{recvschedule}{$r,\recvblock{q}$}
    \For{$e=0,\ldots,q$} $\nextlink{e},\prevlink{e}\gets e-1, e+1$
    \EndFor \Comment Doubly linked, circular list of skips in decreasing order
    \State $\prevlink{q},\nextlink{-1},\prevlink{-1}\gets -1,q,0$
    \State $b\gets\Call{baseblock}{r}$
    \State $\nextlink{\prevlink{b}},\prevlink{\nextlink{b}}\gets \nextlink{b},\prevlink{b}$ \Comment Unlink baseblock index $b$ 
    \State\Call{allblocks}{$p+r,0,p+p,q,0,\recvblock{}$}  
    \For{$k=0,\ldots,q-1$} \Comment Make baseblock $b$ only non-negative block (Condition~\ref{cor:qblocks})
    \If{$\recvblock{k}=q$} $\recvblock{k}\gets b$
    \Else\ $\recvblock{k}\gets\recvblock{k}-q$
    \EndIf
    \EndFor
    \EndProcedure
  \end{algorithmic}
\end{algorithm}

The \textsc{allblocks()} function is now used to compute the receive
schedule for a processor $r,0\leq r<p$ as shown in
Algorithm~\ref{alg:recvschedule}.  The algorithm uses the $q+1$ skips
computed by Algorithm~\ref{alg:circulants} (including $\skips{q}=p$)
and searches for canonical paths to processor $p+r$. In order to
exclude the canonical path leading to $r$ itself with baseblock $b$
(as computed by Algorithm~\ref{alg:baseblock}), $b$ is removed from
the list of skip indices before calling \textsc{allblocks()}. 
The search starts from the largest skip index $e$ with $k=0$ and the
previous processor set to $s=p+p$. Upon return from \textsc{allblocks()},
$\recvblock{k}$ contains $q$ different, non-negative skip indices. In order to
arrive at a schedule where only the index for the baseblock is positive
as required by Theorem~\ref{thm:scheduleuse}, $q$ is subtracted from all
$\recvblock{k}$, except for the $k$ where $\recvblock{k}=p$: This is the
round where processor $r$ receives its baseblock from root $r=p$.

\begin{lemma}
  \label{lem:range}
  A call $\textsc{allblocks}(p+r,0,p+p,q,0,\recvblock{})$ implicitly
  determines $q$ processors $r'_k$ with $r-\skips{k+1}\leq r'_k\leq r-\skips{k}$
  and stores the corresponding $q$ different baseblocks in $\recvblock{k}$.
  Each $\recvblock{k}$ is the largest baseblock in its interval that can
  be found without using the skips of the baseblocks found in the intervals
  before $k$.
\end{lemma}

The existence of such a sequence of processors and baseblocks (skip
indices) is guaranteed by Lemma~\ref{lem:kplus1}.

\begin{proof}
  The recursive calls maintain the invariant that $r'\leq r-\skips{k}$
  before the \textbf{while}-loop, which holds prior to the first call
  where $r'=0$ and $k=0$.  A recursive call is done only if
  $r'+\skips{e}\leq r-\skips{k+1}<r-\skips{k}$ which will maintain the
  invariant for the call. Otherwise, $r-\skips{k}\geq
  r'+\skips{e}>r-\skips{k+1}$. Upon return from a recursive call the
  invariant is rechecked since $k$ can have changed by the recursive
  call. Thus, whenever $r'+\skips{e}$ is accepted as $r'_k$, then
  $r-\skips{k+1}\leq r'_k\leq r-\skips{k}$. Each time a new baseblock
  for $\recvblock{k}$ is found, it is removed from the list of skip
  indices. Therefore, the returned baseblocks will be different.  The
  conditions on the receive blocks that $r'\leq r-\skips{k+1}$ and
  $r'+\skips{e}\geq r-\skips{k+1}$ means that $e$ is the largest
  baseblock in the interval found not using prior baseblocks; no
  further edges inside the interval are explored.
\end{proof}

\begin{lemma}
  Function \textsc{allblocks()} as called from Algorithm~\ref{alg:recvschedule}
  performs at most $q-1$ recursive calls.
\end{lemma}

\begin{proof}
  The \textsc{allblocks()} function maintains the invariant that prior
  to each call, $r'<r-\skips{k}$. After a recursive call, either
  $r'+\skips{e}$ is accepted as $r'_k$ and $e$ removed from the doubly
  linked list never to be considered again, or, if not, then $k$ has
  changed because some other $e$ has been accepted by the recursive
  call.  Therefore, each recursive call causes the elimination of at
  least one skip index $e$, and since there are only $q$ such indices,
  the claim follows since at least one index $e$ in the list will not
  lead to a recursive call (namely, when $r'+\skips{e}=r-\skips{k}$
  which happens at least once when $k=0$).
\end{proof}

\begin{lemma}
  Let $R$ be the number of recursive calls caused by the first
  invocation of \textsc{allblocks()}. Then there are at most $2q+R$
  iterations in total of the \textbf{while}-loop over all calls.
\end{lemma}

\begin{proof}
  All skip indices must be scanned at least once in the
  \textbf{while}-loops of the recursive calls. When for some skip
  index $e$, $r'+\skips{e}>r-\skips{k}$ and no recursive call is
  performed, this index will be scanned again in some earlier, pending
  recursive call (backtrack).  When $r'+\skips{e}$ is not accepted
  after the return from the recursive call, because the invariant that $r'\leq
  r-\skips{k+1}$ is violated, the algorithm backtracks to a point
  where $e$ can be accepted. Thus, the total number of scans is at most
  $2q+R$.
\end{proof}

We have now established that Algorithm~\ref{alg:recvschedule} with the
help of Algorithm~\ref{alg:dfsblocks} in $O(q)$ steps computes
$\recvblock{k},k=0,1,\ldots,q-1$ fulfilling Correctness
Condition~\ref{cor:qblocks} that
$\recvblock{}=\{-1,-2,\ldots,-q\}\setminus\{b-q\}\cup\{b\}$. It 
remains to be shown that Correctness Condition~\ref{cor:recvbefore}
is also fulfilled.  To this end, we define for processor $r$ that
$\sendblock{k}_r = \recvblock{k}_{t^k_r}$ where as in
Algorithm~\ref{alg:broadcast}, $t^k_r=(r+\skips{k})\bmod p$.  With
this definition, the first two correctness conditions from
Section~\ref{sec:broadcasts} are obviously satisfied. It needs to be
shown that $\recvblock{k}_r$ for processor $r$ is either
$\recvblock{j}_{r-\skips{k}}$ for a $j,j<k$ or the baseblock
$b_{r-\skips{k}}$ for processor $r-\skips{k}$.

\begin{lemma}
  For each processor $r, p\leq r<2p$, either
  $\recvblock{k}_r=\recvblock{j}_{r-\skips{k}}$ for some $j,j<k$ or
  $\recvblock{k}_r=b_{r-\skips{k}}$ where $b_{r-\skips{k}}$ is the baseblock
  for processor $r-\skips{k}$.
\end{lemma}
\begin{proof}
  We first note that since
  $(r-\skips{k})-\skips{(k-1)+1}=r-\skips{k}-\skips{k}\geq
  r-\skips{k+1}$, the intermediate processor $r'_k$ chosen by
  processor $r$ in round $k$ is in the range of possible 
  intermediate processors chosen by processor $r-\skips{k}$. The baseblock
  $\recvblock{k}_r$ is the largest in the range $r-\skips{k+1}$ to
  $r-\skips{k}$ not using any prior baseblocks for smaller $k$. If
  $\recvblock{k}_r=b_{r-\skips{k}}$ there is nothing to
  show. Otherwise, $\recvblock{k}_r$ is a largest block in the range excluding
  processor $r-\skips{k}$, and indeed one of the baseblocks found by
  processor $r-\skips{k}$ in the rounds up to $k-1$ since initially
  only skip index $b_{r-\skips{k}}$ is excluded from the search.
\end{proof}

The lemmas together establish our main theorem which states that
correct receive schedules can be computed in optimal time.

\begin{theorem}
  \label{thm:recvsummary}
  For any number of processors $p$, a correct receive schedule for a
  $p$-processor circulant graph with $\ceiling{\log_2 p}$ skips computed by
  Algorithm~\ref{alg:circulants}, fulfilling the four correctness
  conditions, can be computed by Algorithm~\ref{alg:recvschedule} in
  $O(\log p)$ time steps for any processor $r,0\leq r<p$.
\end{theorem}

\subsection{The send schedule}
\label{sec:sendschedule}

A straightforward computation of send schedules from the receive
schedules by for processor $r$ setting
$\sendblock{k}_r=\recvblock{k}_{t^k_r}$ with each $\recvblock{k}$
computed by Algorithm~\ref{alg:recvschedule} will take $O(\log^2 p)$
operations. To reach $O(\log p)$ operations, a different approach is
required. This structural approach can be seen as a generalization of
the well-known, easy case where $p=2^q$ is a power of
two~\cite{JohnssonHo89}. The starting point is captured by
Observation~\ref{obs:doublesend} below which complements
Observation~\ref{obs:double}.

\begin{observation}
  \label{obs:doublesend}
  Let a correct send schedule for $p$ processors be given. A correct
  send schedule for $2p$ processors can be constructed by the
  following modifications.  First, for the large processors $r, p\leq
  r<2p$, copy the send schedule for processor $r-p$, such that
  $\sendblock{k}_r=\sendblock{k}_{r-p}$ for each round $0\leq
  k<q$. Since $q$ is one larger in the $2p$ processor schedule than in
  the $p$ processor schedule, subtract $1$ from all negative blocks in
  the send schedules. For the last round, set $\sendblock{q}_r=b_r$ for
  the small processors $r,0\leq r<p$ where $b_r$ is the baseblock for
  $r$. For the large processors $r,p\leq r<2p$, replace all positive
  send blocks with $-1$ and set $\sendblock{q}=-1$ for the last round.
\end{observation}

The correctness of Observation~\ref{obs:doublesend} can be verified
with the aid of Table~\ref{tab:p9} and Table~\ref{tab:p18} which
doubles $p$ from $9$ to $18$. There is no similarly simple observation
on how to construct a send schedule for $2p-1$ processors from a $p$
processor send schedule. Nevertheless, the send schedule computation
that we adopt here builds on Observation~\ref{obs:doublesend} as much
as possible.

\begin{algorithm}
  \caption{Computing the send schedule for processors $r=0$ (root) and $r,
    0\leq r<p$.}
  \label{alg:sendschedule}
  \begin{algorithmic}
    \Procedure{Sendschedule}{$r,\sendblock{q}$}
    \If{$r=0$}
    \For{$k=0,\ldots,q-1$}\ $\sendblock{k}\gets k$
    \EndFor
    \Else
    \State $b\gets\Call{baseblock}{r}$
    \State $r',c,e \gets r,b,p$
    \For{$k=q-1,\ldots,1$}
    \Comment Obvious invariant: $r'<e$
    \If{$r'<\skips{k}$} \Comment Lower part 
    \If{$r'+\skips{k}<e\vee e<\skips{k-1}\vee (k=1\wedge b>0)$}
    $\sendblock{k}\gets c$
    \Else\Comment Violation 
    \State \Call{recvschedule}{$(r+\skips{k})\bmod p,\ablock[]$} 
    \State $\sendblock{k}\gets\ablock[k]$
    \EndIf
    \If{$e>\skips{k}$} $e\gets \skips{k}$
    \EndIf
    \Else \Comment Upper part, $r'\geq\skips{k}$
    \State $c\gets k-q$
    \If{$k=1\vee r'>\skips{k}\vee e-\skips{k}<\skips{k-1}$}
    $\sendblock{k}\gets c$
    \ElsIf{$r'+\skips{k}>e$} \Comment Violation
    \State \Call{recvschedule}{$(r+\skips{k})\bmod p,\ablock[]$}
    \State $\sendblock{k}\gets\ablock[k]$
    \Else\ $\sendblock{k} \gets c$
    \EndIf
    \State $r',e\gets r'-\skips{k},e-\skips{k}$
    \EndIf
    \EndFor
    \State $\sendblock{0}\gets b-q$
    \EndIf
    \EndProcedure
  \end{algorithmic}
\end{algorithm}

The root processor $r=0$ greedily sends the blocks $0,1,\ldots q-1$ in
the first $q$ communication rounds.  For the send schedule computation
for the other processors, we go through the $q$ communication rounds
starting from $k=q-1$ down to $k=0$. In each round, we divide the
processors into a lower and and upper part, those with $r'<\skips{k}$
and those with $\skips{k}\leq r'$, where $r'$ is a virtual processor
index to be explained. Initially, $r'=r$, and we maintain an invariant
upper bound $e$ on the range of virtual processors, $0\leq
r'<e$. Initially, $e=\skips{q}=p$. To maintain the invariant from
round $k$ to round $k-1$, if $r'$ is in the upper part, both $r'$ and
$e$ are decreased by $\skips{k}$ at the end of round $k$, and if $r'$
is in the lower part, then $e=\skips{k}$ for round $k-1$ unless
already $e<\skips{k}$. For processors in the upper part, $e$ keeps
track of the number of times that $\skips{k}+\skips{k}\neq
\skips{k+1}$, that is, the number of times that $\skips{k+1}$ is odd.
Virtual processors in the lower part aim to send the same block $c$ in
round $k$ as they already sent in round $k+1$. Initially, for these
processors, $c=b_r$ where $b_r$ is the baseblock for processor $r$ as
found by Algorithm~\ref{alg:baseblock}. Virtual processors in the
upper part aim to send block $c=k-q$ in round $k$. These processors
start with block $c=(q-1)-q=-1$ and decrease $c$ by $1$ in each round
as indicated by Observation~\ref{obs:doublesend}. The complete
algorithm maintaining the invariant on the virtual processor indices
is shown as Algorithm~\ref{alg:sendschedule}.

If in round $k$, following Observation~\ref{obs:doublesend}, the $r'$
for processor $r$ is in the lower part, $r'<\skips{k}$, the processors
$r'+\skips{k}<e$ have not yet received block $c$, so $c$ is to be sent
if $r'+\skips{k}<e$. Otherwise, it we do not know which block
processor $(r+\skips{k})\bmod p$ is missing, so in that case the
receive block (as computed by Algorithm~\ref{alg:recvschedule}) for
round $k$ for processor $(r+\skips{k})\bmod p$ is taken as the block
to send. This takes $O(q)$ steps, and is called a \emph{violation}. If
there are more than a constant number of such violations for some
processor $r$, a logarithmic number of operations in total for that
processor cannot be guaranteed. If $e$ is very small, $e<\skips{k-1}$,
processor $(r+\skips{k})\bmod p$ will not have received $c$ and this
block can therefore be sent.

If instead $r'$ is in the upper part for round $k$, then only the
processor with $r'=\skips{k}$ may have to use the receive block for
processor $(r+\skips{k})\bmod p$ as the block to send. 

\begin{theorem}
  \label{thm:sendsummary}
  Algorithm~\ref{alg:sendschedule} computes for any $r,0\leq r<p$ a
  send schedule in $O(\log p)$ operations.
\end{theorem}
\begin{proof}
  The loop performs $q-1$ iterations.  Iterations that are not
  violations take constant time. We will show that there can be only a
  constant number of violations for $r$, namely at most four.
  Each violation takes $O(\log p)$ steps by the receive schedule
  Theorem~\ref{thm:recvsummary}. Therefore, the send schedule
  computation takes $\Theta(\log p)$ steps.

  Assume that a violation happens at iteration $k$ for $r'$, under
  what conditions can a violation happen again for a smaller $k$?  If
  $r'$ is in the upper part, $r'\geq\skips{k}$, the violation can
  happen only for $r'=\skips{k}$ by the first condition (indeed,
  $\sendblock{k}=c$ when $r'>\skips{k}$). For the next, lower
  iteration, $r'$ is decreased by $\skips{k}$ to $r'=0$ and will
  remain so for all remaining iterations, thus it will never again
  hold that $r'\geq\skips{k}$. We therefore only have to consider
  violations for the lower part case where $r'\leq\skips{k}$.  A
  violation happens if $r'+\skips{k}\geq e$. If $e\geq\skips{k}$,
  $e=\skips{k}$ for the next, lower iteration. Thus, a violation
  cannot happen in iteration $k-1$ if $r'<\skips{k-1}$. If
  $r'\geq\skips{k-1}$, $r'$ will be in the upper part for iteration
  $k-1$ with $e=\skips{k}$ and so $r'+\skips{k-1}>\skips{k}=e$ can
  happen only for $r'=\skips{k-1}$ (and by Observation~\ref{obs:skips}
  only if $\skips{k}$ is odd, which is then taken care of in the upper
  part for iteration $k-1$). Thus, as long as $e\geq\skips{k}$ there
  can be at most two violations for $r$. The difficult cases are when
  $e<\skips{k}$ which will cause a violation for $\skips{k-1}\leq r'$.
  These $r'$ will be in the upper part for the next iteration $k-1$,
  but there could possibly be a violation for some later, smaller
  $k$. At which iterations can it happen that $\skips{k-1}\leq
  e\leq\skips{k}$?

  The smallest possible $e$ in iteration $k$ is
  \begin{eqnarray*}
    p-\sum_{i=k+1}^{q-1}\skips{i} & = &
    p -(\sum_{i=0}^{q-1}\skips{i}-\sum_{i=0}^{k}\skips{i}) \\
    & = & p-(\skips{q}+q-1 -(\skips{k+1}+k)) \\
    & = & p-(p+q-1-\skips{k+1}-k) \\
    & = & \skips{k+1}+1-q+k
  \end{eqnarray*}
  by the upper bound of Lemma~\ref{lem:bounds}.

  A violation always happens if $\skips{k-1}\leq e\leq\skips{k}$.
  For the smallest
  possible $e$, we have $\skips{k+1}+1-q+k\leq\skips{k}$ implying
  $\skips{k+1}-\skips{k}\leq q-1-k$ and with
  Observation~\ref{obs:skips} that $\skips{k}\leq q-k$ and with
  Observation~\ref{obs:skips} again that $\skips{k-1}\leq
  \frac{q+1-k}{2}$.  Likewise, for a violation to happen,
  $\skips{k-1}\leq e$ gives $\skips{k-1}\leq\skips{k+1}+1-q+k$. Using
  Observation~\ref{obs:skips} twice, this implies $\skips{k-1}\leq
  4\skips{k-1}+1-q+k$ and therefore $\frac{q-1-k}{3}\leq \skips{k-1}$.
  Since $\skips{k}$ are roughly halving with decreasing $k$, both
  inequalities are fulfilled for a very small number of $k$ on the order
  of $O(\log\log p)$. This is small enough, that such a violation cannot
  happen again.

  Counting, there are at most four times that a violation can happen for
  $r$.
\end{proof}

Theorem~\ref{thm:recvsummary} and Theorem~\ref{thm:sendsummary} are
crucial, also for practical use of the schedule computations. To
verify the implementations, finite, exhaustive proof for $p$ up to
some millions shows that the computed receive and send schedules are
correct and consistent according to the four conditions in
Section~\ref{sec:broadcasts}, that the number of violations claimed in
Theorem~\ref{thm:sendsummary} is indeed at most $4$ and that the
stated invariants all hold (see the appendix).  Table~\ref{tab:p17}
shows receive and send schedules as computed by the algorithms for
$p=17$ processors (not a power of two). There are, for instance, send
schedule violations in round $k=2$ for processor $r=3$ and in round
$k=3$ for processor $r=8$.

\section{Empirical Results}
\label{sec:empirics}

As experimental evidence, the algorithms for computing receive and
send schedules in $O(\log p)$ time steps are being used in new
implementations of \mpibcast, \mpiallgatherv, \mpireduce,
\mpireducescatterblock and \mpireducescatter. All implementations are
available from the author. We show here preliminary results with
\mpibcast (Algorithm~\ref{alg:broadcast}) and \mpireduce on a large
HPC system and \mpiallgatherv with irregular problems on a small
system. We compare our new implementations against a native MPI
library implementation.  The best number of blocks $n$ leading to the
smallest broadcast time is chosen based on a linear cost model. For
\mpibcast, the size of the blocks is chosen as $F\sqrt{m/\ceiling{\log
    p}}$ for a constant $F$ chosen experimentally. For \mpiallgatherv,
the number of blocks to be used is chosen as $\sqrt{m\ceiling{\log
    p}}/G$ for another, experimentally determined constant $G$.  The
constants $F$ and $G$ depend on context, system, and MPI
library. Finding a best $n$ in practice is a highly interesting
problem outside the scope of this work.

\begin{figure}
  \begin{center}
    \includegraphics[width=.32\linewidth]{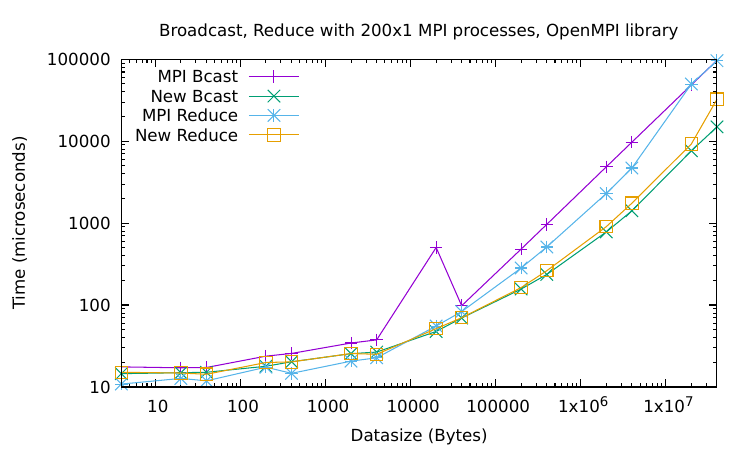}
    \includegraphics[width=.32\linewidth]{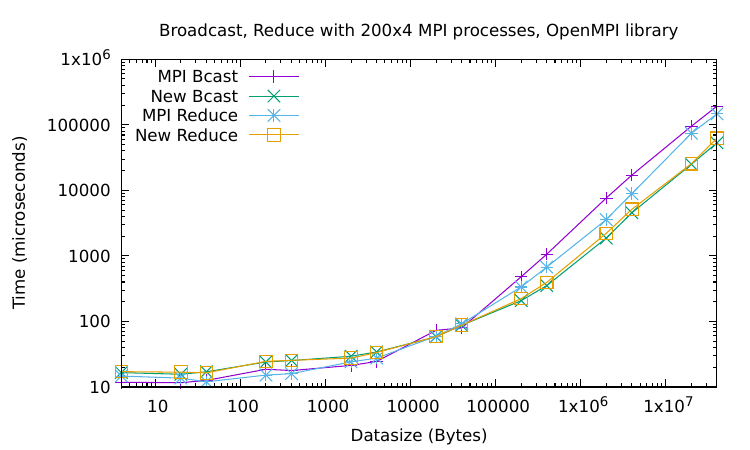}
    \includegraphics[width=.32\linewidth]{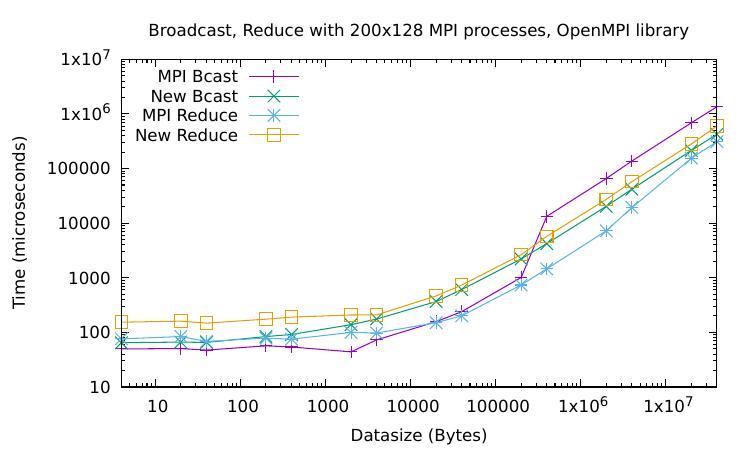}
  \end{center}
  \caption{Broadcast and reduce results, native versus new, with the
    \vegampi library with $p=200\times 1, p=200\times 4$ and $p=200\times 128$
    MPI processes.  The constant factor $F$ for the size of the blocks
    has been chosen as $F=70$. The MPI datatype is \mpiint.}
  \label{fig:bcastreducevega}
\end{figure}

The large system is the VEGA system in Slowenia, which is comprised of
768 CPU based compute nodes each two AMD EPYC 7H12 Rome 64-core CPUs
running at 2.6 GHz for a total of 128 cores per node.  We have run
with 200 nodes and in Figure~\ref{fig:bcastreducevega} show results
with $p=200\times 1$, $p=200\times 4$ and $p=200\times 128$ MPI
processes. Our new \mpibcast and \mpireduce implementations are run
against the native \vegampi library implementation. Even though our
implementation is not specifically tuned, and in particular not
adapted to the hierarchical structure of the system, our
implementation is faster than the native MPI library by a factor of
more than $4$ and $3$, respectively for the one and four processes per
node configurations for the large problem sizes. Even with full nodes,
our (non-hierarchical, fully-connected) implementation of \mpibcast is
a factor of 3 faster than the native implementation. Part of the
reason for this must be poor implementations in the \vegampi library.
We will focus on fine tuning and hierarchical implementation of our
new algorithms elsewhere.

The small system is a $36\times 32$ processor cluster with 36~dual socket
compute nodes, each with two Intel(R) Xeon(R) Gold 6130F 16-core
processors. The nodes are interconnected via dual Intel Omnipath
interconnects each with a bandwidth of 100 GigaBytes/s. The
implementations and benchmarks were compiled with \gcc with the
\texttt{-O3} option.  The MPI library used is \hydraopenmpi.

\begin{figure*}
  \begin{center}
    \includegraphics[width=.32\linewidth]{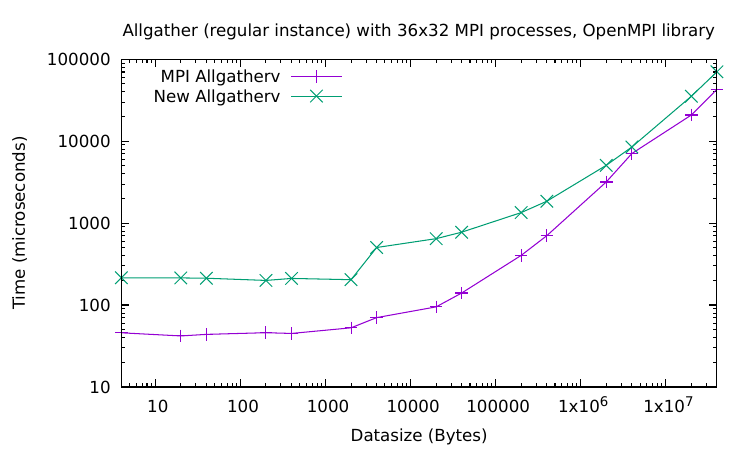}
    \includegraphics[width=.32\linewidth]{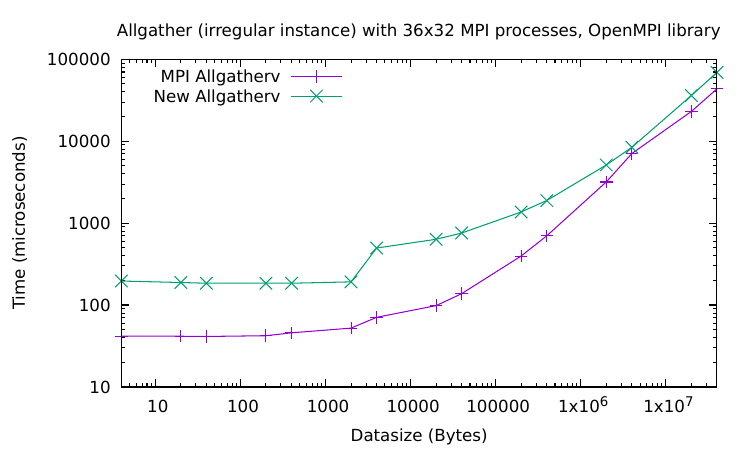}
    \includegraphics[width=.32\linewidth]{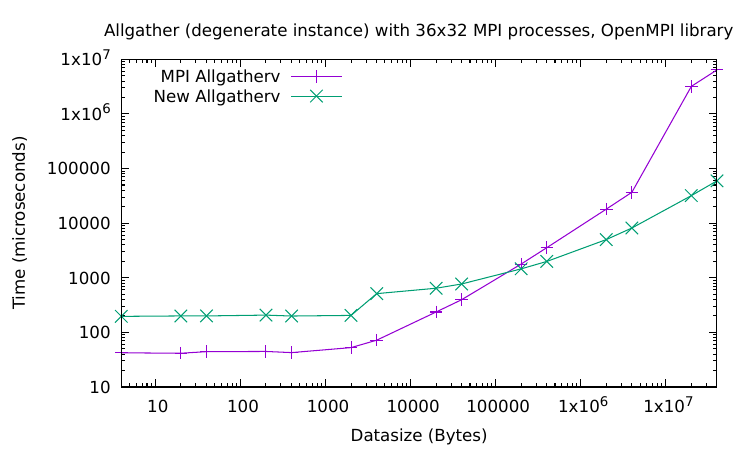}
  \end{center}
  \caption{Irregular all-broadcast (\mpiallgatherv) results, native
    versus new, with the \hydraopenmpi library with $p=36\times 32$
    MPI processes and different types of input problems (regular,
    irregular, degenerate).  The constant factor $G$ for the number of
    blocks has been chosen as $G=40$. The MPI datatype is \mpiint.}
  \label{fig:allgatopenmpi}
\end{figure*}

The results for \mpiallgatherv are shown in
Figure~\ref{fig:allgatopenmpi} for $p=36\times 32$ MPI processes and
different types of input problems (see the appendix for a detailed
algorithm description). The \emph{regular problem} divides the given
input size $m$ roughly evenly over the processes in chunks of $m/p$
elements. The \emph{irregular problem} divides the input in chunks of
size roughly $(i\bmod 3) m/p$ for process $i=0,1,\ldots, p-1$. The
\emph{degenerate problem} has one process contribute the full input of
size $m$ and all other process no input elements. For the degenerate
problem, the performance of the \hydraopenmpi baseline library indeed
degenerates and is a factor of close to 100 slower than the new
implementation, where the running time is largely independent of the
problem type. The running time of the new \mpiallgatherv
implementation is in the ballpark of \mpibcast for the same total
problem size. This is highly attractive for MPI libraries and not
guaranteed by other implementations. The overheads for the schedule
computations are one to two orders of magnitude smaller than those
shown in~\cite{Traff22:bcastba}.

\section{Summary}

We showed that round-optimal broadcast schedules on fully connected,
one-ported, fully bidirectional $p$-processor systems can indeed be
computed in $O(\log p)$ time steps per processor. This is a surprising
and highly non-trivial result that affirmatively answers the long
standing open questions posed
in~\cite{Traff22:bcastba,Traff22:bcast,Traff08:optibcast}. One reason
for the difficulty and the awkward proofs is that there is no obvious
counterpart to Observation~\ref{obs:double} and
Observation~\ref{obs:doublesend} for odd, $2p-1$, number of
processors. We furthermore observed that such schedules that rely on
regular, circulant graph communication patterns likewise give rise to
round-optimal, pipelined algorithms for the related all-broadcast,
reduction and all-reduction operations. We outlined how to use the
schedules for implementation of the MPI collectives \mpibcast,
\mpiallgatherv, \mpireduce, \mpireducescatterblock and
\mpireducescatter, and gave promising, initial experiments.  A more
careful evaluation of these implementations, also in versions that are
more suitable to systems with hierarchical, non-homogeneous
communication systems~\cite{Traff20:mpidecomp} is ongoing and should
be found elsewhere.

\bibliographystyle{plainurl}
\bibliography{traff,parallel} 

\appendix
\section{All-Broadcast}

\begin{algorithm}
  \caption{The $n$-block all-to-all broadcast algorithm for processor
    $r,0\leq r<p$ for data in the arrays $\mathtt{bufs}[j],0\leq
    j<p$. The count $x$ is the number of empty first rounds.  Blocks
    smaller than $0$ are neither sent nor received, and for blocks
    larger than $n-1$, block $n-1$ is sent and received instead.}
  \label{alg:irregallgather}
  \begin{algorithmic}
    \Procedure{AllBroadcast}{$\mathtt{bufs}[p][n]$}
    \For{$j=0,1,\ldots,p-1$}
    \State $r'\gets (r-j+p)\bmod p$
    \State$\Call{recvschedule}{r',\mathtt{recvblocks}[j][]}$
    \EndFor
    \For{$j=0,1,\ldots,p-1$}
    \For{$k=0,1,\ldots,q-1$}
    \State $f\gets (j-\skips{k}+p)\bmod p$
    \State$\mathtt{sendblocks}[j][k]\gets\mathtt{recvblocks}[f][k]$
    \EndFor
    \EndFor
    \State $x\gets (q-(n-1)\bmod q)\bmod q$
    \Comment Virtual rounds, only when $x>0$
    \For{$j=0,1,\ldots, p-1$}
    \For{$i=0,1,\ldots,q-1$}
    \State $\mathtt{recvblocks}[j][i]\gets\mathtt{recvblocks}[j][i]-x$
    \State $\mathtt{sendblocks}[j][i]\gets\mathtt{sendblocks}[j][i]-x$
    \If{$i<x$} \Comment Virtual rounds before $x$ already done
    \State $\mathtt{recvblocks}[j][i]\gets\mathtt{recvblocks}[j][i]+q$
    \State $\mathtt{sendblocks}[j][i]\gets\mathtt{sendblocks}[j][i]+q$
    \EndIf
    \EndFor
    \EndFor
    \For{$i=x,x+1,\ldots,(n+q-1+x)-1$}
    \State $k\gets i\bmod q$
    \State $t, f\gets (r+\skips{k})\bmod p,(r-\skips{k}+p)\bmod p$ 
    \State $j'\gets 0$
    \For{$j=0,1,\ldots,p-1$} \Comment Pack
    \If{$j\neq t$}
    \State $\mathtt{tempin}[j']\gets\mathtt{bufs}[j][\mathtt{sendblocks}[j][k]]$
    \State $j'\gets j'+1$
    \EndIf
    \State $\mathtt{sendblocks}[j][k]\gets \mathtt{sendblocks}[j][k]+q$ \EndFor

    \State $\bidirec{\mathtt{tempin},t}{\mathtt{tempout},f}$

    \State $j'\gets 0$
    \For{$j=0,1,\ldots,p-1$} \Comment Unpack
    \If{$j\neq r$}
    \State $\mathtt{bufs}[j][\mathtt{recvblocks}[j][k]],\gets\mathtt{tempout}[j']$
    \State $j'\gets j'+1$
    \EndIf
    \State $\mathtt{recvblocks}[j][k]\gets \mathtt{recvblocks}[j][k]+q$
    \EndFor
    \EndFor
    \EndProcedure
\end{algorithmic}
\end{algorithm}

The explicit send and receive schedules can be used for all-broadcast as
shown as Algorithm~\ref{alg:irregallgather}. In this
problem, each processor has $n$ blocks of data to be broadcast to all
other processors. The $n$ blocks for each processor $r$ are as for the
broadcast algorithm assumed to be roughly of the same size, but blocks
from different processors may be of different size as long as the same
number of $n$ blocks are to be broadcast from each processor. The
algorithm can therefore handle the irregular case where different
processors have different amounts of data to be broadcast as long as
each divides its data into $n$ roughly equal-sized blocks.  Due to the
fully symmetric, circulant graph communication pattern, this can be
done by doing the $p$ broadcasts for all $p$ processors $r, 0\leq r<p$
simultaneously, in each communication step combining blocks for all
processors into a single message.  The blocks for processor $j$ are
assumed to be stored in the buffer array $\mathtt{bufs}[j][]$
indexed by block numbers from $0$ to $n-1$.  Initially, processor $r$
contributes its $n$ blocks from $\mathtt{bufs}[r][]$. The task is
to fill all other blocks $\mathtt{bufs}[j][]$ for $j\neq r$.  Each
processor $r$ computes a receive schedule $\mathtt{recvblocks}[j]$ for
each other processor as root processor $j,0\leq j<p$ which is the
receive schedule for $r'=(r-j+p)\bmod p$.
Before each communication operation, blocks for all processors
$j,0\leq j<p$ are packed consecutively into a temporary buffer
\texttt{tempin}, except the block for the to-processor $t^k$ for the
communication round. This processor is the root for that block, and
already has the corresponding block. After communication, blocks from
all processors are unpacked from the temporary buffer \texttt{tempout}
into the $\mathtt{bufs}[j][]$ arrays for all $j,0\leq j<p$ except for
$j=r$: A processor does not receive blocks that it already has. As in
the broadcast algorithm in Algorithm~\ref{alg:broadcast}, it is
assumed that the packing and unpacking will not pack for negative
block indices, and that indices larger than $n-1$ are taken as
$n-1$. Also packing and unpacking blocks for processors not
contributing any data (as can be the case for highly irregular
applications of all-to-all broadcast) shall be entirely skipped (not
shown in Algorithm~\ref{alg:irregallgather}), so that the total time
spent in packing and unpacking per processor over all communication
rounds is bounded by the total size of all $\mathtt{bufs}[j][], j\neq
t^k$ and $\mathtt{bufs}[j][], j\neq r$.

\section{Overheads}

\begin{table}
  \caption{Timings of old $O(\log^3 p)$ and new $O(\log p)$ time step
    receive and send schedule algorithms for different ranges of
    processors $p$. Receive and send schedules are computed for all
    processors $0\leq r<p$ for all $p$ in the given ranges. The
    expected running times are thus bounded by $O(p\log^3 p)$ (old)
    and $O(p \log p)$ (new) time, respectively.  Times are in seconds
    and measured with the \texttt{clock()} function. We also estimate
    the average time spent per processor for computing its send and
    receive schedules of $\ceiling{\log_2 p}$ entries. This is done by
    measuring for each $p$ the total time for the schedule
    computation, dividing by $p$ and averaging over all $p$ in the
    range. These times are in micro seconds.}
  \label{tab:runtimes}
  \begin{small}
  \begin{center}
    \begin{tabular}{crrrr}
      Proc.\ range $p$ &
      \multicolumn{2}{c}{Total Time (seconds)} &
      \multicolumn{2}{c}{Per proc.\ ($\mu$seconds)} \\
      & $O(p\log^3 p)$ & $O(p\log p)$ & $O(\log^3 p)$ & $O(\log p)$ \\
      \toprule
      $[1,17\,000]$ & 443.8 & 50.0 & 2.769 & 0.334 \\
      $[16\,000,33\,000]$ & 1567.2 & 152.8 & 3.763 & 0.370 \\
      $[64,000\,73\,000]$ & 3206.0 & 282.6 & 5.187 & 0.454 \\
      $[131\,000,140\,000]$ & 7595.0 & 653.2 & 6.226 & 0.534 \\
      $[262\,000,267\,000]$ & 9579.4 & 726.6 & 7.242 & 0.548 \\
      $[524\,000,529\,000]$ & 21580.2 & 1492.9 & 8.196 & 0.566 \\
      $[1\,048\,000,1\,050\,000]$ & 18934.3 & 1083.8 & 9.024 & 0.516 \\
      $[2\,097\,000,2\,099\,000]$ & 44714.9.0 & 2554.6 & 10.656 & 0.608 \\
      \bottomrule
    \end{tabular}
  \end{center}
  \end{small}
\end{table}

To demonstrate the practical impact of the improvement from $O(\log^3
p)$ time steps (per processor) which was the bound given
in~\cite{Traff22:bcast} to $O(\log p)$ steps per processor as shown
here, we run the two algorithms for different ranges of processors
$p$. For each $p$ in range, we compute both receive and send schedules
for all processors $r, 0\leq r<p$, and thus expect total running times
bounded by $O(p\log^3 p)$ and $O(p\log p)$, respectively.  These
runtimes in seconds, gathered on a standard workstation with an Intel
Xeon E3-1225 CPU at 3.3GHz and measured with the \texttt{clock()}
function from the \texttt{time.h} C library, are shown in
Table~\ref{tab:runtimes}. The timings include both the receive and the
send schedule computations, but exclude the time for verifying the
correctness of the schedules, which has also been performed up to some
$p\geq 2M, M=2^{20}$ (and for a range of $100\,000$ processors around
$16M$), including verifying the bounds on the number of recursive
calls from Theorem~\ref{thm:recvsummary} and the number of
violations from Theorem~\ref{thm:sendsummary}. When receive and send
schedules have been computed for all processors $r,0\leq r<p$,
verifying the four correctness conditions from
Section~\ref{sec:broadcasts} can obviously be done in $O(p\log p)$
time steps. The difference between the old and the new implementations
is significant, from close to a factor of 10 to significantly more
than a factor of 10. However, the difference is not by a factor of
$\log^2 p$ as would be expected from the derived upper bounds.  This
is explained by the fact that the old send schedule implementations
employ some improvements beyond the trivial computation from the
receive schedules which makes the complexity closer to $O(\log^2 p)$.
We also give a coarse estimate of the time spent per processor
by measuring for each $p$ in the given range the time for the schedule
computations for all $p$ processors, dividing this by $p$ and
averaging over all $p$ in the given range. This is indicative of the
overhead for the \textsc{recvblock()} and \textsc{sendblock()}
computations in the implementations of Algorithm~\ref{alg:broadcast}
and Algorithm~\ref{alg:irregallgather}. These times (in microseconds)
are listed as columns $O(\log^3 p)$ and $O(\log p)$, respectively.
The difference by a factor of $10$ and more is slowly increasing with
$\log_2 p$.

\end{document}